\documentclass[10pt,twocolumn,english,conference]{IEEEtran}
\usepackage{mathptmx} 
\usepackage{cite}
\usepackage{multicol}
\usepackage{amsmath}
\usepackage{subeqnarray}
\usepackage{cases}
\usepackage{framed}
\usepackage{amsfonts}
\usepackage{amssymb}
\usepackage{bm}
\usepackage{graphicx}
\usepackage{epstopdf}
\usepackage{import}
\usepackage{subfig}
\usepackage[english]{babel}
\usepackage{color}%
\usepackage{algorithm}
\usepackage{url}
\usepackage[belowskip=-6pt,aboveskip=0pt]{caption}
\usepackage{amsthm}
\usepackage{amsbsy}

\newtheorem*{thm*}{Theorem}

\theoremstyle{remark}

\theoremstyle{definition}

\newtheorem*{definition*}{Definition}

\newtheorem*{ex*}{Example}

\providecommand{\U}[1]{\protect\rule{.1in}{.1in}}

\newcommand{\E}{\mathbb{E}}

\newcommand{\xb}{\mathbf{x}}
\newcommand{\yb}{\mathbf{y}}

\newcommand{\Rt}{\mathbb{R}^2}

\usepackage[
left=1.02in,
right=1.02in,
top=0.72in,
bottom=1.04in
 ]{geometry}

\IEEEoverridecommandlockouts \IEEEpubid{\makebox[\columnwidth]{ 978-1-5386-3531-5/17/\$31.00~\copyright~2017 IEEE \hfill} \hspace{\columnsep}\makebox[\columnwidth]{ }}

\begin{document}

\title{\LARGE A Traffic Model for Machine-Type Communications Using Spatial Point Processes}


\author{\IEEEauthorblockN{Henning Thomsen, Carles Navarro Manch\'on, Bernard Henri Fleury}
\IEEEauthorblockA{Department of Electronic Systems, Aalborg University, Denmark\\
Email: \{ht, cnm, fleury\}@es.aau.dk} }

\maketitle


\begin{abstract}
A source traffic model for machine-to-machine communications is presented in this paper. We consider a model in which devices operate in a regular mode until they are triggered into an alarm mode by an alarm event. The positions of devices and events are modeled by means of Poisson point processes, where the generated traffic by a given device depends on its position and event positions. We first consider the case where devices and events are static and devices generate traffic according to a Bernoulli process, where we derive the total rate from the devices at the base station. We then extend the model by defining a two-state Markov chain for each device, which allows for devices to stay in alarm mode for a geometrically distributed holding time. The temporal characteristics of this model are analyzed via the autocovariance function, where the effect of event density and mean holding time are shown.
\end{abstract}

%



\section{Introduction}\label{sec:Introduction}
Machine-type communications (MTC), also called machine-to-machine (M2M) communications, is a type of transmission which does not directly involve human intervention. Examples include smart metering, telemetry and sensors~\cite{shafiq2013large}. This type of transmission is projected to be an important part of fifth generation (5G) communication systems~\cite{boccardi2014five}.

It is expected that the number of Machine-type Devices (MTDs) will be at least two orders of magnitude higher than the number of human-type users in a given cell~\cite{yang2014m2m}. Also, MTC traffic from several devices can be correlated in time and space, since e.g. sensors can react to external events such as a fire or traffic accident, which have a spatial characteristic. This sort of behaviour can result in traffic spikes or bursts which can put stress on the receiver (e.g. a base station (BS)) since MTC traffic is usually uplink-dominated. It is therefore desirable to have tractable models that can be used to assess the impact of this type of traffic on the BS.

In this paper, we formulate a traffic model based on tools from stochastic geometry that takes into account the spatial aspect of devices and events. Motivated by the observation that traffic from individual MTDs is often small, while the total traffic at the BS can be large~\cite{shafiq2013large}, we derive analytical expressions for the expected total rate of the devices. It is assumed that each device can be in two modes, regular and alarm, depending on the proximity of the device location to the event locations. We also derive a Markov chain model and provide an analytical expression which is a good approximation for the expected total rate when the time spent in the alarm state is low.
%

\section{Background and Related Works}\label{sec:BackgroundandRelatedWorks}
Models for M2M traffic can broadly be categorized into two groups, source models and aggregated models~\cite{centenaro2015study}. In a source model, each source has its own set of tunable parameters which control its individual traffic characteristics. In an aggregated traffic model all sources have a common set of parameters. A source model thus allows for more fine-grained control; however this comes at a higher cost in computational complexity because of a higher number of parameters.

Traffic modeling for MTC has been considered by the 3rd Generation Partnership Project (3GPP) in~\cite{3GPPTR37868}. They propose two models, one accounting for uncoordinated arrivals and one accounting for coordinated arrivals. In the former, the packet arrivals are specified to follow a uniform distribution in $[0,T_u]$, where $T_u$ is $60$ sec. The latter model uses a scaled Beta distribution over $[0,T_b]$, where $T_b$ is $10$ sec, with parameters $\alpha=3$ and $\beta=4$. In~\cite{madueno2016reliable}, other values of $\alpha$ and $\beta$ are derived numerically for different types of alarm events.

Modeling of human-type traffic (such as telephone traffic) is usually done by means of Poisson processes (PPs). However, PPs are not well suited for modeling coordinated traffic, which is expected to play a large role in MTC applications. To address this, many papers consider Markov modulated Poisson processes (MMPPs). In such a model, a source can be in a number of states $\mathcal{S} = \{s_1, \ldots, s_m\}$ and the state transitions are governed by an irreducible Markov chain. When the process is in state $s_i \in \mathcal{S}$, it generates Poisson-distributed traffic with rate $\lambda_i$. A variant of this is the Markov modulated Bernoulli process (MMBP), where Bernoulli-distributed traffic with state-dependent probability $p_i$ is generated~\cite{ibe2013markov}. Another model used especially in Internet traffic modeling is the Markovian arrival process (MAP). A MAP can be in a number of states $s_i \in \mathcal{S}$, also termed \emph{phases}~\cite{ibe2013markov}. The process can then either enter another phase with rate $\gamma$, or, with rate $\lambda$, enter another phase and generate a packet. Contrary to the MMBP and MMPP, where the interarrival times of packets are geometrically and exponentially distributed respectively, the MAP allows for more general interarrival time distributions.

To take into account possible dependency between MTDs,~\cite{laner2013traffic} considers a model where each device can be in regular or alarm state. Two state-transition matrices $\mathbf{P}_{\mathrm{U}}$ and $\mathbf{P}_{\mathrm{C}}$ are defined, the former representing uncoordinated behaviour and the other, coordinated behaviour. The model is a source model where each device has its own transition matrix which is a convex combination of $\mathbf{P}_{\mathrm{U}}$ and $\mathbf{P}_{\mathrm{C}}$. The coordination between devices is modeled via a parameter $\delta_i$ for each device $i$, as well as a global background process $\theta(t)$, which takes values uniformly in $[0,1]$. It is shown that the model has a lower computational complexity compared to a full source model. This model is extended in~\cite{centenaro2015study} and used for evaluating a random access scheme supporting human-type and MTD traffic.

A spatio-temporal traffic model for MTDs using queueing theory and stochastic geometry is analyzed in~\cite{gharbieh2017spatiotemporal}. They compare three different access protocols for MTC and derive success probabilities, average queueing delays and waiting times. Stability regions for the three protocols are given, which can be used for scalability assessment. Compared to this work, our model allows for devices to be triggered by external events.

Reference~\cite{madueno2016reliable} formulates a spatial correlation model of MTDs using the IEEE802.11ah protocol. As in other works, each device can be in two states, regular or alarm, with transition probabilities described by the model in~\cite{laner2013traffic}. In~\cite{madueno2016reliable}, they define the background process $\theta(t)$ to be a function of alarm onset and velocity, as well as location. They give a numerical expression for the total rate, albeit not in closed form.

A two-state traffic model for automotive M2M applications based on a coupled MAP is proposed in~\cite{grigoreva2017coupled}. Here, any car can trigger an event; then other cars go into event mode if the distance to the car triggering the event is less than a specified threshold. They show that this model better captures burstiness than a PP model


In the present paper, we use the theory of spatial Poisson point processes (PPPs) to model the position of devices and event epicenters. We also define a function which models the influence of events on device traffic. Even though spatial modeling has been considered in~\cite{madueno2016reliable, grigoreva2017coupled}, our use of PPPs allows us to obtain tractable analytical expressions for the total rate, and thus avoids computationally expensive simulations. Also, we include the temporal aspect in the model via defining a Markov chain for each device, derive an approximation of the expected total rate, and show via numerical simulations that this approximation is quite close. We study the temporal correlation of the traffic in the Markov chain model via the autocovariance function. Even though our model is a source model, the theory of PPPs can still be applied and provides accurate results.

\section{Proposed Model}\label{sec:ProposedModel}
Let $C$ be the coverage disk of a BS (located at the origin of the disk) of radius $R$ and let $N$ devices be deployed randomly and independently in $C$.
Each MTD can be in one of two states: regular (R) and alarm (A). Time is slotted. In timeslot $k$, MTD $i$ generates traffic with rate $R_i(k)$, depending on its current state. In the regular state, it generates traffic with rate $R_{\mathrm{R},i}$, and in alarm state, the rate is $R_{\mathrm{A},i}$. In this work, all devices are of the same type, i.e. $R_{\mathrm{R},i} = R_{\mathrm{R}}$ and $R_{\mathrm{A},i} = R_{\mathrm{A}}$, $\forall i=1,\ldots,N$.

\subsection{Modeling the Device Locations and Event Epicenters}\label{sub:DeviceModeling}
The MTDs are deployed in the coverage disk of the BS according to a homogeneous PPP $\Phi_{\mathrm{M}}$ with density $\lambda_{\mathrm{M}}$. The event epicenters are represented by a homogeneous PPP $\Phi_{\mathrm{E}}$ with density $\lambda_{\mathrm{E}}$ in the Euclidean plane. The processes $\Phi_{\mathrm{M}}$ and $\Phi_{\mathrm{E}}$ are assumed independent.

We choose to use PPPs because typical nodes can be reasonably assumed to be randomly deployed in the plane, in particular since we are not targeting a specific application. Furthermore, this choice allows for analytical tractability and can serve as a good indicator of performance.

Fig.~\ref{fig:ScenarioEventDevicesBaseStation} shows an example of a deployment of devices in a cell. Event epicenters are indicated with red crosses. Note that events from outside the BS coverage area can influence devices inside it. In the figure, three devices are in alarm mode (transmission indicated with solid red arrows) and two devices are in regular mode (transmission indicated with dashed blue arrows).

\begin{figure}[t]
	\centering
		\includegraphics[width=0.7\linewidth]{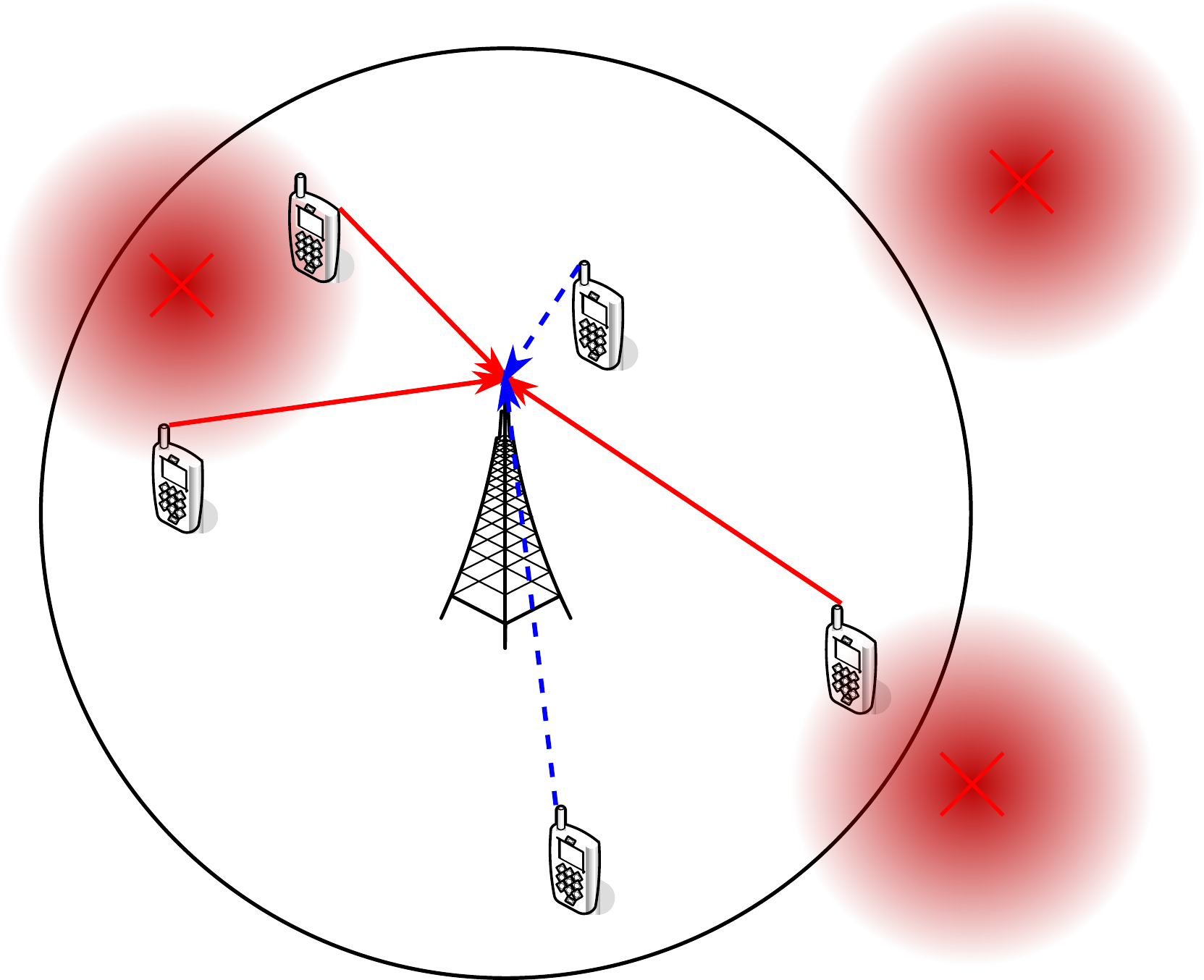}
	\caption{Conceptual deployment of devices in a cell, with a BS at the center. Event epicenters are indicated by red crosses. The influence of an event is coded in the shade intensity: The darker the shade, the stronger the influence.}
	\label{fig:ScenarioEventDevicesBaseStation}
\end{figure}

\subsection{Modeling Alarm Probabilities}\label{sub:ModelingAlarmProbabilities}

In order to capture the effect of a given event on a device, we define a function with gives the probability $p$ that an alarm is triggered in a device at location $\xb$ by an event with epicenter at location $\yb$ in the Euclidean plane $\Rt$. This function depends on the distance between the two locations. The definition is given below.
\begin{definition*}[Alarm Triggering Probability Function (ATPF)]\label{def:ECF}
An \emph{ATPF} is a function $f$ defined as
\begin{equation}\label{eqn:ATPFdefinition}
	f : [0,\infty) \to [0,1], \hspace{1cm} d \mapsto f(d) = p,
\end{equation}
having a finite first moment.
\end{definition*}
In applications, $f(d)$ is typically non-increasing to represent a decaying influence of events on devices as the distance $d$ increases.

\subsection{Modeling Discrete Rates}\label{sub:ModelingDiscreteRates}
Let $E_{\xb\yb}$ be the event that a device at location $\xb$ is triggered into alarm mode by an event with epicenter at $\yb$.
Let $\overline{E_{\xb\yb}}$ be the complement of $E_{\xb\yb}$ and let $p_{\xb\yb} = \Pr\{ E_{\xb\yb}\}$ be the probability of $E_{\xb\yb}$. The probability of this device being in alarm mode is
\begin{align}
p_\xb &= \Pr\{ \text{Device at $\xb$ is  triggered by at least one event} \}  \nonumber \\
&= 1 - \Pr\{ \text{No event triggers device at $\xb$} \} \\
&= 1 - \Pr \left\{ \bigcap_{\xb \in \Phi_{\mathrm{E}}} \overline{E_{\xb\yb}} \right\} = 1 - \prod_{\yb \in \Phi_{\mathrm{E}}} \Pr \left\{ \overline{E_{\xb\yb}} \right\} \label{eqn:probAlarmDerivation} \\
&= 1 - \prod_{\yb \in \Phi_{\mathrm{E}}} (1 - p_{\xb\yb}) \label{eqn:FinalExpressionPi},
\end{align}
where in~\eqref{eqn:probAlarmDerivation}, we used the assumption that alarms resulting from different events are independent. Note that $0 \le 1-f(d_{\xb\yb}) \le 1$, where $d_{\xb\yb} = \Vert \xb - \yb \Vert$  and $p_{\xb\yb} = f(d_{\xb\yb})$.

The state of a device at location $\xb$, $S_\xb$ is a random variable which depends on $\Phi_{\mathrm{M}}$ and $\Phi_{\mathrm{E}}$. The statistical dependence between $S_\xb$, the corresponding rate $R_\xb$ and $\Phi_{\mathrm{M}}$ and $\Phi_{\mathrm{E}}$ is represented as a directed graphical model, as shown in Fig.~\ref{fig:StateEventDependence}. In the figure, the blue box (called a plate, see e.g.~\cite[Ch. 8]{bishop2006pattern}) around $\xb, S_\xb$ and $R_\xb$ indicates $\vert \Phi_{\mathrm{M}} \vert$ instances of each of them, where $\vert \Phi_{\mathrm{M}} \vert$ is the cardinality of $\Phi_{\mathrm{M}}$. Similarly, the right plate around $\yb$ indicates an infinite but countable number of instances of $\yb$.
%
%
Note that $S_\xb$ depends on an infinite number of parameters, since $\Phi_{\mathrm{E}}$ contains an infinite number of points almost surely. However, $S_\xb$ is still well-defined, as long as the ATPF is chosen according to the definition and $R < \infty$.

Since each device at location $\xb$ has its own probability $p_\xb$ of going into alarm mode, the present model can be considered a source model. Also, the model allows for correlated behaviour, since for two devices at locations $\xb$ and $\xb'$ respectively, where $\xb \neq \xb'$, both $S_\xb$ and $S_{\xb'}$ depend on all points of $\Phi_{\mathrm{E}}$.

We consider two ways of modeling the states $S_\xb(k)$, with slot index $k$. Below, we model them as i.i.d. Bernoulli random variables, later in Sec.~\ref{sec:MarkovchainAnalysis} we model them as a Markov chain. For the former case,
\begin{equation}\label{eqn:BernoulliModelCases}
S_\xb(k) = \begin{cases} \text{Regular} & \text{with prob.}\; 1-p_\xb \\ \text{Alarm} & \text{with prob.} \;p_\xb \end{cases},
\end{equation}
with $p_\xb$ given in~\eqref{eqn:FinalExpressionPi}. Then, each realization of $\Phi_{\mathrm{E}}$ and $\Phi_{\mathrm{M}}$ results in state sequences $\{S_\xb(k)\}_{k=1}^\infty$, which are i.i.d. Bernoulli processes.
\def\svgwidth{0.55\columnwidth}
\begin{figure}[t]
	\centering
		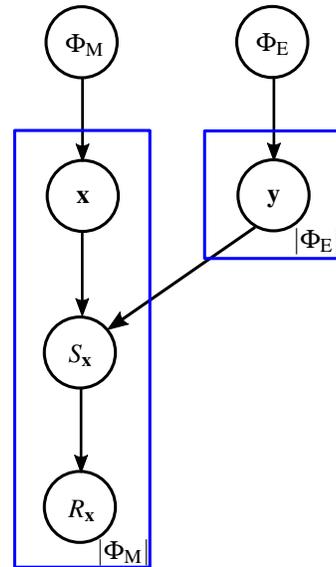
	\caption{Dependence of the rate of the devices on the two PPPs, illustrated as a directed graphical model.}
	\label{fig:StateEventDependence}
\end{figure}

\section{Derivation of the Expected Total Rate}\label{sec:ExpectedAggegrateRateDerivation}

\subsection{Bernoulli Traffic Generation}\label{sub:BernoulliTrafficGeneration}
In this subsection, we assume that each device generates traffic according to a Bernoulli process.

The expected rate of the device at location $\xb$ is
\begin{align}
R_\xb &= p_\xb R_{\mathrm{A}} + (1-p_\xb)R_{\mathrm{R}}.
\end{align}
By substituting the expression of $p_\xb$ in~\eqref{eqn:FinalExpressionPi}, we get
\begin{align}
R_\xb &= p_\xb R_{\mathrm{A}} + (1-p_\xb)R_{\mathrm{R}} \\
&= \left( 1 - \prod_{\yb \in \Phi_{\mathrm{E}}} (1 - p_{\xb\yb}) \right) R_{\mathrm{A}} + \prod_{\yb \in \Phi_{\mathrm{E}}} (1 - p_{\xb\yb}) R_{\mathrm{R}} \\
&= R_{\mathrm{A}} + ( R_{\mathrm{R}} - R_{\mathrm{A}} ) \prod_{\yb \in \Phi_{\mathrm{E}}} (1 - p_{\xb\yb}).
\end{align}

The total rate $R_{\mathrm{T}} = \sum_{\xb \in \Phi_{\mathrm{M}}} R_\xb$ of all devices is then
\begin{equation}\label{eqn:totalRateNotExpectation}
\sum_{\xb \in \Phi_{\mathrm{M}}} \left( R_{\mathrm{A}} + ( R_{\mathrm{R}} - R_{\mathrm{A}} ) \prod_{\yb \in \Phi_{\mathrm{E}}} (1 - p_{\xb\yb}) \right).
\end{equation}
The goal is now to obtain an expression for the total rate averaged over $\Phi_{\mathrm{M}}$ and $\Phi_{\mathrm{E}}$. This expression is given in the following theorem:

\begin{thm*}\label{thm:totalRateAnalytical}
Let $\Phi_{\mathrm{M}}$ and $\Phi_{\mathrm{E}}$ be independent homogeneous PPPs with densities $\lambda_{\mathrm{M}}$ and $\lambda_{\mathrm{E}}$ respectively, and $f$ be a valid ATPF. Let $C$ be the circular coverage disk of the BS, and $R$ its radius. Then the expected total rate averaged over $\Phi_{\mathrm{M}}$ and $\Phi_{\mathrm{E}}$ is
\begin{align}\label{eqn:totalRateAnalytical}
\overline{R_{\mathrm{T}}} &\triangleq \E \left[R_{\mathrm{T}} \right] =  \lambda_{\mathrm{M}} \pi R^2 \cdot \nonumber \\ 
& \left( R_{\mathrm{A}} + (R_{\mathrm{R}} - R_{\mathrm{A}})  \int_C \exp \left( -2\pi\lambda_{\mathrm{E}} \int_{0}^\infty f(r)r \mathrm{d}r \right) \right).
\end{align}
\end{thm*}
\begin{proof}
Taking the expectation over $\Phi_{\mathrm{M}}$ and $\Phi_{\mathrm{E}}$ in~\eqref{eqn:totalRateNotExpectation}, we get 
\begin{align}
&\E \left[ \sum_{\xb \in \Phi_{\mathrm{M}}} \left( R_{\mathrm{A}} + ( R_{\mathrm{R}} - R_{\mathrm{A}} ) \prod_{\yb \in \Phi_{\mathrm{E}}} (1 - p_{\xb\yb}) \right) \right] \\ \nonumber
&= \E \left[ \sum_{\xb \in \Phi_{\mathrm{M}}} R_{\mathrm{A}} \right] + \E \left[ \sum_{\xb \in \Phi_{\mathrm{M}}} (R_{\mathrm{R}} - R_{\mathrm{A}}) \prod_{\yb \in \Phi_{\mathrm{E}}}(1-p_{\xb\yb}) \right] \\ 
&= \lambda_{\mathrm{M}} \int_C R_{\mathrm{A}} + (R_{\mathrm{R}} - R_{\mathrm{A}}) \E \left[ \sum_{\xb \in \Phi_{\mathrm{M}}} \prod_{\yb \in \Phi_{\mathrm{E}}}(1-p_{\xb\yb}) \right] \label{eqn:CampbellThm1} \\ \nonumber
&= \lambda_{\mathrm{M}} \pi R^2 R_{\mathrm{A}} \\ & \hspace{1cm}+ (R_{\mathrm{R}} - R_{\mathrm{A}}) \E \left[ \sum_{\xb \in \Phi_{\mathrm{M}}} \E \left[ \prod_{\yb \in \Phi_{\mathrm{E}}} (1-p_{\xb\yb}) \right] \right] \label{eqn:LastTermInnerExpectation},
\end{align}
where in~\eqref{eqn:CampbellThm1} Campbell's Theorem~\cite[Thm. 4.1]{chiu2013stochastic} was used for a PPP with density $\lambda_{\mathrm{M}}$ over a circular disk $C$ of radius $R$.

To compute the inner expectation in the second term of~\eqref{eqn:LastTermInnerExpectation}, we use the probability generating functional (PGFL) of a PPP with density $\lambda_{\mathrm{E}}$~\cite[p.125]{chiu2013stochastic}. After changing to polar coordinates, we get for each summand in~\eqref{eqn:LastTermInnerExpectation}
\begin{equation}
\E \left[ \prod_{\yb \in \Phi_{\mathrm{E}}} (1 - f(d_{\xb\yb})) \right] = \exp \left( -2\pi \lambda_{\mathrm{E}} \int_0^\infty f(r) r \mathrm{d}r \right).
\end{equation}
Then, the outer expectation in~\eqref{eqn:LastTermInnerExpectation} equals
\begin{align}
\E \left[ \sum_{\xb \in \Phi_{\mathrm{M}}} \exp \left( -2\pi \lambda_{\mathrm{E}} \int_0^\infty f(r) r \mathrm{d}r \right) \right]. \label{eqn:secondEquationCampbell}
\end{align}
Using Campbell's Theorem, the expectation in~\eqref{eqn:secondEquationCampbell} can be expressed as
\begin{align}
&\lambda_{\mathrm{M}} \int_C \exp \left( -2\pi \lambda_{\mathrm{E}} \int_0^\infty f(r) r \mathrm{d}r \right) \\ &= \lambda_{\mathrm{M}} \pi R^2 \exp \left( -2\pi \lambda_{\mathrm{E}} \int_0^\infty f(r) r \mathrm{d}r \right).
\end{align}
Combining the above results and rearranging yields
\begin{align}
\lambda_{\mathrm{M}} \pi R^2 \left( R_{\mathrm{A}} + (R_{\mathrm{R}} - R_{\mathrm{A}})\exp \left( -2\pi \lambda_{\mathrm{E}} \int_0^\infty f(r) r \mathrm{d}r \right) \right). \label{eqn:FinalExpressiontotalRate}
\end{align}
\end{proof}

\begin{ex*}[Exponential ATPF]\label{ex:expATPF}
Consider an exponentially decaying function of the distance $d$:
\begin{equation}\label{eqn:ExpATPF}
f(d) = \exp(-d).
\end{equation}
It is easy to see that~\eqref{eqn:ExpATPF} satisfies the conditions in the Definition of an ATPF. In this case, $p_{\xb\yb} = \exp ( -d_{\xb\yb})$. The expected total rate is
\begin{equation}\label{eqn:rateExpATPF}
\overline{R_{\mathrm{T}}} = \lambda_{\mathrm{M}} \pi R^2 (R_{\mathrm{A}} + (R_{\mathrm{R}} - R_{\mathrm{A}}) \exp( -2\pi \lambda_{\mathrm{E}} )),
\end{equation}
which follows since $\int_0^\infty \exp(-r) r \mathrm{d} r = 1$.
Using~\eqref{eqn:rateExpATPF}, we can see that $\lambda_{\mathrm{E}} \to 0$ (i.e. no event) implies $\overline{R_{\mathrm{T}}} = \lambda_{\mathrm{M}} \pi R^2 R_{\mathrm{R}}$ (only regular traffic), while $\lambda_{\mathrm{E}} \to \infty$ implies $\overline{R_{\mathrm{T}}} = \lambda_{\mathrm{M}}\pi R^2 R_{\mathrm{A}}$ (only alarm traffic), as expected.
\end{ex*}

\section{Markov chain Analysis}\label{sec:MarkovchainAnalysis}
In this section, we extend the model, enabling each device to stay in alarm mode for a duration geometrically distributed with parameter $q$.

We characterize the state of a device at location $\xb \in \Phi_{\mathrm{M}}$ at a given time $k$ as a Markov process, i.e.
\begin{equation}
\Pr \{S_\xb(k) \mid S_\xb(k-1), \ldots, S_\xb(0) \} = \Pr \{ S_\xb(k)\mid S_\xb(k-1) \}
\end{equation}
where $S_\xb(k)$ can take two different values, Regular and Alarm. The state transition diagram is shown in Fig.~\ref{fig:StateTransitionDiagram}, and the state transition matrix is
\begin{equation}\label{eqn:StateTransitionMatrixSourceModel}
\mathbf{P}_\xb = \begin{bmatrix} 1-p_\xb& p_\xb \\ 1-q & q\end{bmatrix}.
\end{equation}
This Markov chain is ergodic; it has a unique steady-state probability vector $\pi_\xb = [\pi_{\xb,A}, \; \pi_{\xb,R}]$, where $\pi_{\xb,A}$ ($\pi_{\xb,R}$) is the probability of alarm (regular) state. The vector $\pi_\xb$ can be found by solving the system of linear equations
\begin{equation}
	\pi_\xb \mathbf{P}_\xb = \pi_\xb,
\end{equation}
together with the normalizing condition $\pi_{\xb,A} + \pi_{\xb,R} = 1$:
\begin{equation}\label{eqn:SteadyStateProbs}
	\pi_{\xb,A} = \frac{p_\xb}{1+p_\xb-q}, \hspace{1cm} \pi_{\xb,R} = \frac{1-q}{1+p_\xb-q}.
\end{equation}

We assume that the initial state $S_\xb(0)$ is distributed according to the above steady-state distribution.
\def\svgwidth{0.85\columnwidth}
\begin{figure}[t]
	\centering
		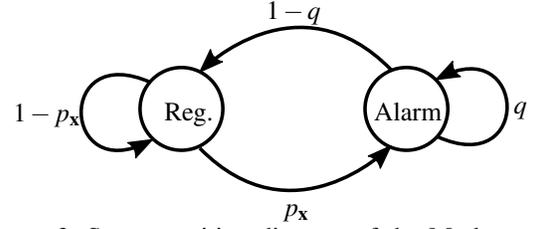
	\caption{State transition diagram of the Markov chain model describing the temporal behaviour of the $i$th device.}
	\label{fig:StateTransitionDiagram}
\end{figure}
%

Note that by setting $q = p_\xb$, we recover the model defined in~\eqref{eqn:BernoulliModelCases}. Also, by replacing each $p_\xb$ in~\eqref{eqn:SteadyStateProbs} with $\overline{p} \triangleq \E \left[  p_\xb \right]$, we can approximate the total rate $R(k) = \sum_{\xb \in \Phi_{\mathrm{M}}}R_\xb(k)$ of the Markov chain model in a low complexity manner:
\begin{align}
	\widetilde{R_{\mathrm{T}}} &= \E \left[ \sum_{\xb \in \Phi_{\mathrm{M}}} \left( \frac{\overline{p}}{1+\overline{p}-q} R_{\mathrm{A}} + \frac{1-q}{1+\overline{p}-q} R_{\mathrm{R}} \right) \right] \\
	&= \lambda_{\mathrm{M}} \pi R^2 \left( \frac{\overline{p}}{1+\overline{p}-q} R_{\mathrm{A}} + \frac{1-q}{1+\overline{p}-q} R_{\mathrm{R}} \right).\label{eqn:totalRateApproxMarkovchain}
\end{align}
We comment on this approximation in Sec.~\ref{sec:NumericalResults}. Note that by setting $q=0$ and multiplying $p_\xb$ by a factor $\theta(t)$ in the Markov chain, we recover the models of~\cite{laner2013traffic, centenaro2015study}. 

\subsection{Temporal Characteristics of the Model}\label{sub:TemporalCorrelationoftheModel}
We study the temporal characteristics of the Markov chain model through the autocovariance function (ACF) of the total rate.

The reason for introducting the ACF is that we want to see the effect of the parameters $\lambda_{\mathrm{E}}$, $\lambda_{\mathrm{M}}$ and $q$ on the long-term temporal dependence of the model. We suspect that by letting $q \to 1$, the autocovariance function will decay slowly, since devices triggered by an alarm will stay in the alarm mode for a longer period of time, implying temporally correlated traffic. The ACF of $R(k)$ is defined as
\begin{equation}\label{eqn:ACF}
	C_{RR}(k,k') = \frac{1}{\sigma_{\mathrm{R}}^2}\E \left[ (R(k)-\mu_{\mathrm{R}})(R(k')-\mu_{\mathrm{R}}) \right],
\end{equation}
where $k, k' = 1,2,\ldots$ and $\mu_{\mathrm{R}}$ and $\sigma_{\mathrm{R}}^2$ are the mean and variance of $R(k)$ respectively. The expectation is taken over the distribution of $R(k)$.

Note that we use the autocovariance function instead of the autocorrelation, because we want to compare the temporal dependence for various settings of the  parameters $\lambda_{\mathrm{E}}$, $\lambda_{\mathrm{M}}$ and $q$ on the same scale.

We estimate the ACF of $R(k)$ via Monte Carlo (MC) simulations as follows. For any trial, the unbiased sample autocovariance
\begin{equation}\label{eqn:sampleACF}
	\widehat{C}_{RR}(k) = \frac{1}{(N_S-k)\sigma_{\mathrm{R}}^2}\sum_{i=1}^{N_S-k-1}(R(i)-\mu_{\mathrm{R}})(R(i+k)-\mu_{\mathrm{R}}),
\end{equation}
is computed.
Here $N_S$ is the number of time samples. Note that in our case, ~\eqref{eqn:sampleACF} is conditioned on realizations of $\Phi_{\mathrm{E}}$ and $\Phi_{\mathrm{M}}$. Also, the mean and variance are replaced by their respective sample estimates. Averaging $\widehat{C}_{RR}(k)$ over the PPP realizations via MC simulations, we get $\overline{C_{RR}}(k)$. 
In the next section, the role of this quantity on the traffic will be discussed. The derivation of the ACF is left for future work.

\section{Numerical Results}\label{sec:NumericalResults}
In this section, we show some quantitative results illustrating the behavioural features of the Bernoulli and Markov chain models, using MC simulations done in Matlab. In the simulations, devices and events are deployed independently and traffic is generated either using the Bernoulli or Markov chain models. The resulting total rate is then averaged over event and device realizations. We use the exponential ATPF~\eqref{eqn:ExpATPF} in all simulations. The simulation parameters are listed in Tab.~\ref{tab:SimulationParameters}. They are used unless otherwise specified.

\begin{table}[b]
	\caption{Simulation parameters.}
	\centering
			\begin{tabular}{ c l c }
		  	\hline
  			Parameter & Description & Value\\
				\hline
				$s$ & Observation window size & $100$ m \\
				$\lambda_{\mathrm{M}}$ & Device density & $10^{-1}$ ${\text{nodes}}/{\text{m}^2}$ \\
				$\lambda_{\mathrm{E}}$ & Event density & $10,\ldots, 10^{-5}$ ${\text{events}}/{\text{m}^2}$ \\
				$R_{\mathrm{A}}$ & Rate in alarm mode& $1$ $\text{packet}/\text{slot}$ \\
				$R_{\mathrm{R}}$ & Rate in regular mode& $0.01$ $\text{packet}/\text{slot}$ \\
				$R$ & Cell radius & $20$ m \\
				\hline
			\end{tabular}
	\label{tab:SimulationParameters}
\end{table}

\subsection{Bernoulli Process Model}\label{sub:BernoulliProcessModel}
Fig.~\ref{fig:rateAnaSimBern} shows a comparison between simulation and analytical results for the Bernoulli model. We see that the theoretical and simulation curves match, which validates our model. Also, we see that for $\lambda_{\mathrm{E}} \to 0$ only regular traffic occurs and $\lambda_{\mathrm{E}} \to \infty$ implies that only alarm traffic is received at the BS. Three regions are identified, two saturation regions (when $\lambda_{\mathrm{E}} \to 0$ and $\lambda_{\mathrm{E}} \to \infty$) and a transition region.

\begin{figure}[t]
	\centering
		\includegraphics[width=0.75\linewidth]{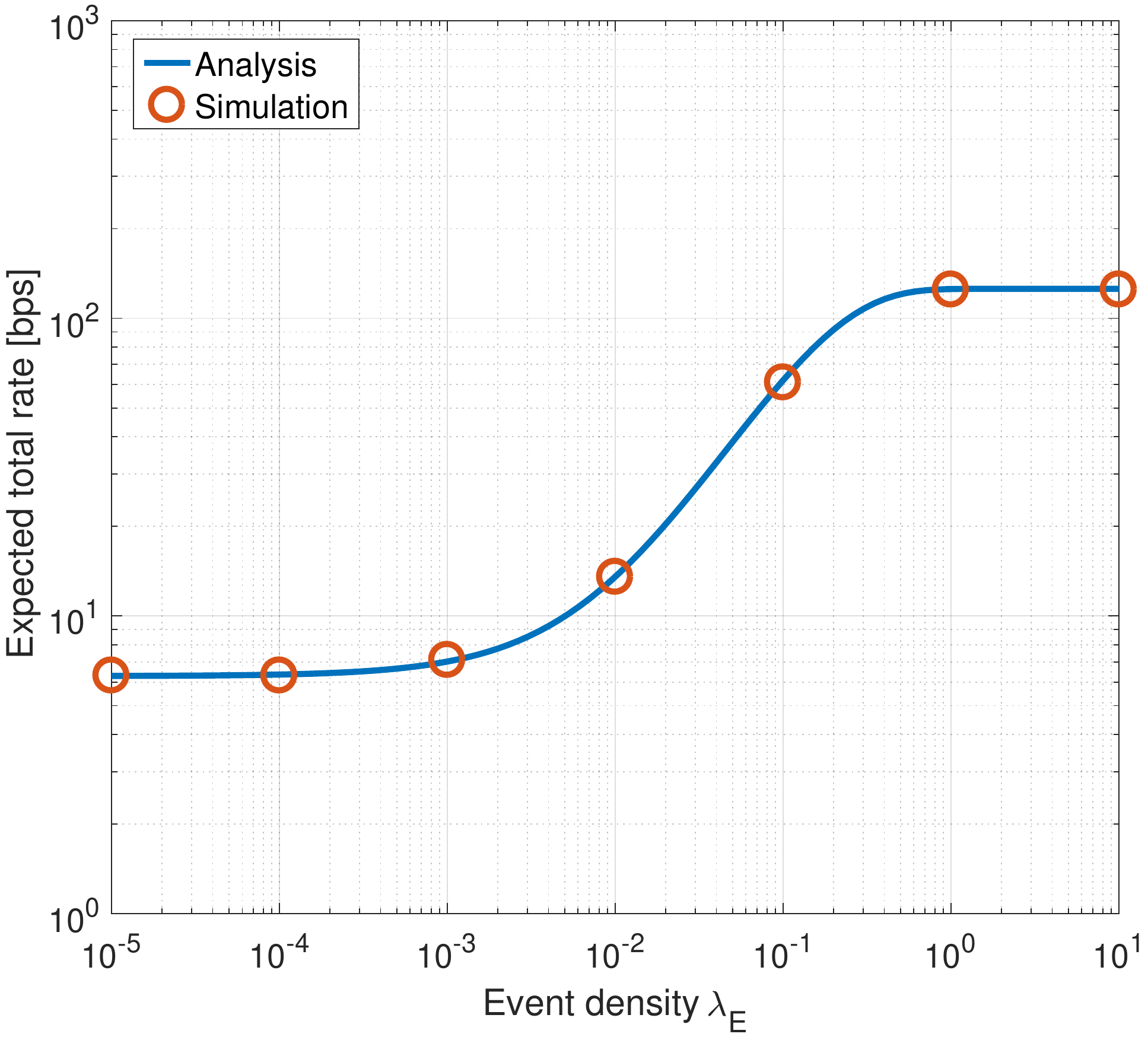}
	\caption{Expected total rate for Bernoulli traffic.}
	\label{fig:rateAnaSimBern}
\end{figure}

\subsection{Markov Chain Model}\label{sub:MarkovchainModel}
The Markov chain model in Sec.~\ref{sec:MarkovchainAnalysis} is numerically evaluated in this subsection. We compare the total rate computed from the simulations against the approximation $\widetilde{R_{\mathrm{T}}}$ in~\eqref{eqn:totalRateApproxMarkovchain} and the analytical expression for $\overline{R_{\mathrm{T}}}$ in~\eqref{eqn:totalRateAnalytical}. We chose to include $\overline{R_{\mathrm{T}}}$ since we want to compare the two expressions and identify the numerical values of $q$ for which they are close.

In Fig.~\ref{fig:rateAnaSimBernMCvaryq} the expected total rate of the Markov chain model is shown for two values of $\lambda_E$. It can be seen that for the Bernoulli model, the rate is constant, since it does not depend on $q$. As for the Markov chain model, the rate obtained from the simulation is always upper bounded by the steady-state approximation, where the latter is obtained from~\eqref{eqn:totalRateApproxMarkovchain}. Further, these two curves have the same shape. The approximation can be used as an upper bound on the total rate at the BS; however the bound is not tight.

\begin{figure}[t]
	\centering
		\includegraphics[width=0.7\linewidth]{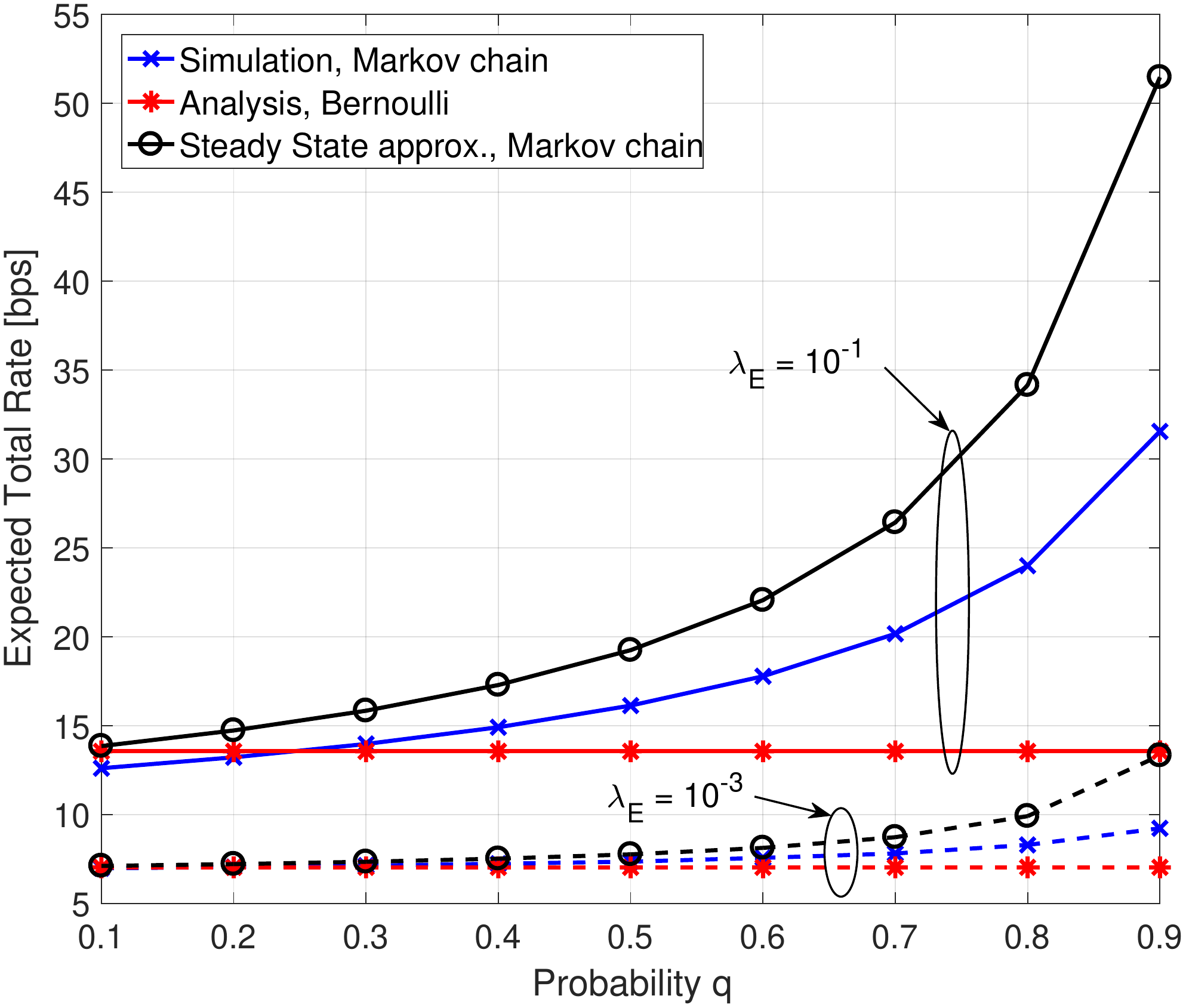}
	\caption{Expected total rates for Bernoulli traffic and for the Markov chain model; $\lambda_{\mathrm{M}} = 10^{-1}$.}
	\label{fig:rateAnaSimBernMCvaryq}
\end{figure}

\begin{figure}[t]
	\centering
		\includegraphics[width=0.7\linewidth]{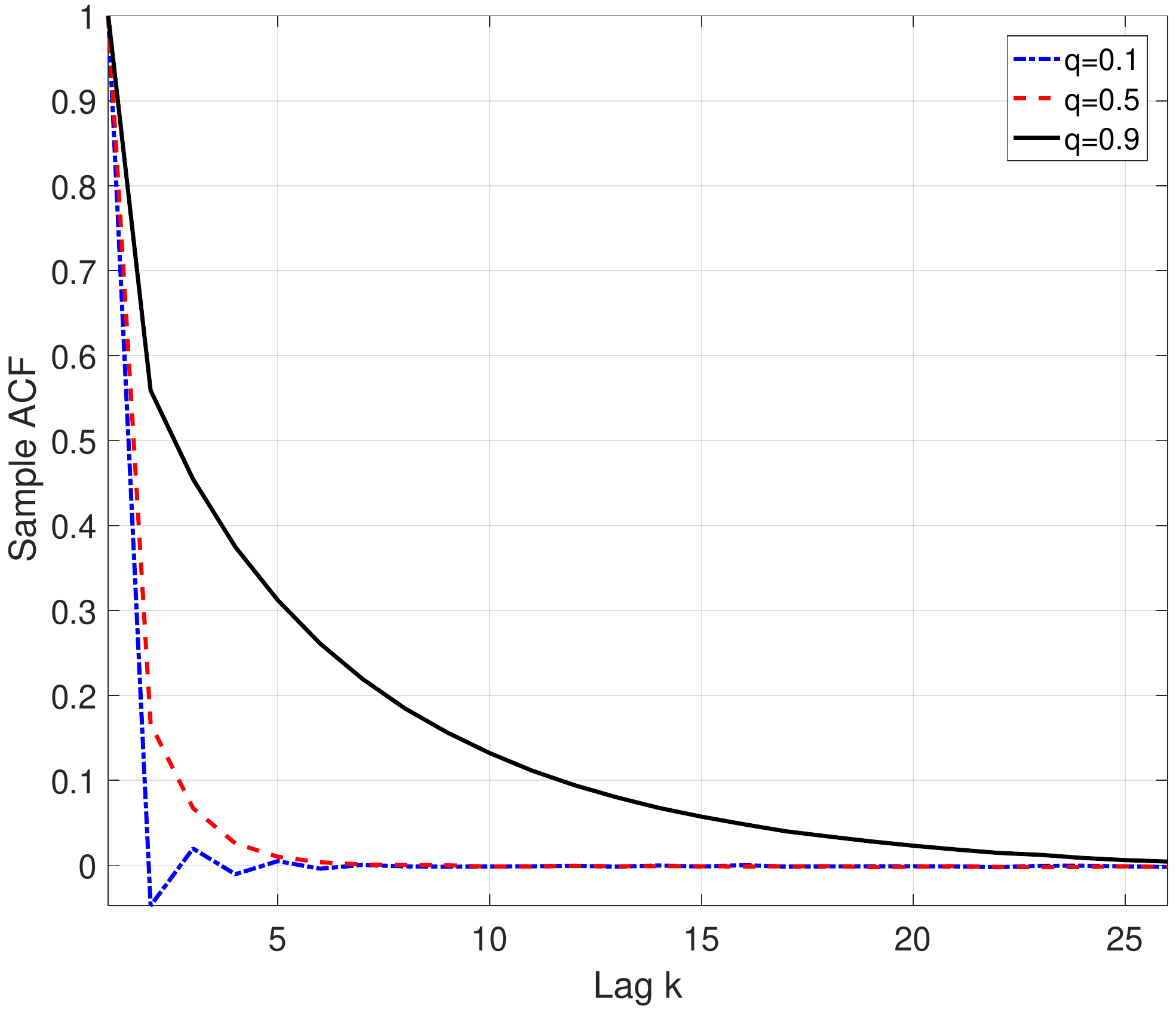}
	\caption{Sample ACF of the total rate for the Markov chain model; $\lambda_{\mathrm{E}} = 10^{-2}$, $\lambda_{\mathrm{M}} = 10^{-2}$.}
	\label{fig:autocovarianceE001M001}
\end{figure}

The sample autocovariance $\overline{C_{RR}}(k)$ is shown in Fig.~\ref{fig:autocovarianceE001M001}. Focusing first on the case $q=0.1$, we see that the ACF almost behaves like the one from an i.i.d. Bernoulli process, in that it has short memory. This is because the $p_\xb$s are low (through the relatively low densities of the PPPs) and because of the relatively small value of $q$. Increasing the value of $q$, see the dashed red (for $q=0.5$) and solid black (for $q=0.9$) curves, increases the memory, i.e. the total rate at a given time $k$ is correlated with many past values. This is because the Markov chain now stays in the alarm state longer when it enters that state.

Hence, the introduction of the additional parameter $q$ in the Markov chain model allows for tuning the temporal correlation of the individual rate processes of the MTDs and, as a result, that of the total rate process.

\section{Conclusion and Future Work}\label{sec:Conclusion}
In this paper, a MTC traffic model based on spatial point processes was given. We also provided a verification of the model via simulations and derived an expression for the total rate of the devices. This expression is very fast to compute, compared to a full simulation, especially for high device and event densities. A Markov chain model was also defined. It was seen that approximating the probability $p_\xb$ by $\overline{p}$ for all devices, gives a good approximation of the total rate for low values of $q$.
The parameter $q$ of the Markov chain model enables tuning of the model to suit various MTC applications and alarm reporting strategies. The models can be used to assess network performance in terms of the total rate and to see the impact of the event density. Furthermore, the Markov chain model, along with the tunable parameter $q$, can be used to study the impact of events on the temporal correlation of the total traffic at the BS.
Also, even though our model is a source model, the expected total rate is still fast to compute, thanks to the machinery of PPPs. As future work, we plan to include time-dependency in the event modeling, as well as to consider several types of MTDs and events. Also, we plan to model other types of sources, such as Poisson and periodic sources, and study the variance of the total traffic, along with its impact on network performance.

\section{Acknowledgements}\label{sec:Acknowledgements}
This work has been supported by the cooperative project VIRTUOSO, partially funded by Innovation Fund Denmark.

\bibliographystyle{ieeetr}
\bibliography{biblio}

\end{document}